\documentclass[3p,final]{elsarticle}

\usepackage{amssymb}
 \usepackage{amsthm}
\newtheorem{Theorem}{Theorem}[section]

\newtheorem{lemma}[Theorem]{Lemma}
\newtheorem{Corollary}[Theorem]{Corollary}

\usepackage{subfloat}
\usepackage{hyperref}
\newtheorem{Observation}[Theorem]{Observation}
\newtheorem{Claim}[Theorem]{Claim}
\usepackage{natbib}

\journal{arXiv}

\begin{document}

\begin{frontmatter}
\title{A local characterization for perfect plane near-triangulations}


\author[mymainaddress]{Sameera M Salam}
\ead{shemi.nazir@gmail.com}
\author[mysecondaryaddress]{Jasine Babu}
\ead{jasine@iitpkd.ac.in}
\author[mymainaddress]{K Murali Krishnan}
\ead{kmurali@nitc.ac.in}

\address[mymainaddress]{Department of Computer Science and Engineering, National Institute of Technology Calicut, Kerala, India 673601}
\address[mysecondaryaddress]{Department of Computer Science and Engineering, Indian Institute of Technology, Palakkad, Kerala, India 678557}

\begin{abstract}
We derive a local criterion for a plane near-triangulated graph
to be perfect.  It is shown that a plane near-triangulated graph is perfect if and only if it does not contain either a vertex, 
an edge or a triangle, the neighbourhood of which has an odd hole as its boundary. The characterization leads to an $O(n^{2})$ 
algorithm for checking perfectness of plane near-triangulations. 
\end{abstract}


\end{frontmatter}

\section{Introduction}
A plane embedding of a planar graph $G$ is said to be a plane near-triangulation if 
all its faces, except possibly the exterior face, are triangles.  
It was known even before the strong perfect graph theorem (\citet{chudnovsky2006strong}) 
that a planar graph is not perfect if and only if
it contains an induced odd hole (\citet{tucker1973strong}). 
Algorithmic recognition of planar perfect graphs was subsequently studied by \citet{hsu1987recognizing} who discovered a method to determine whether
a given planar graph of $n$ vertices is perfect in $O(n^{3})$ time.   Later, ~\citet{cornuejols2003polynomial} discovered an algorithm 
to recognize perfect graphs in $O(n^{9})$ time, using the strong perfect graph theorem. 
    
Though structural characterizations for perfect plane triangulations 
were attempted in the literature (see for example, \citet{benchetrit2015h}), 
a local characterization for perfect plane triangulations (or plane near-triangulations) does not seem to be known.  
An attempt in this direction was initiated by \citet{SalamCWKS19}.  In this work, we extend their results to obtain a local
characterization for a plane (near-) triangulated graph to be perfect. The characterization leads to an $O(n^{2})$ algorithm for checking perfectness
of a plane near-triangulation.  No quadratic time algorithm seems to be known
in the literature for testing perfectness of plane triangulations.  

If a plane near-triangulation $G$ contains a cut vertex or an edge separator, we can split $G$ into
two induced subgraphs such that $G$ is perfect if and only if each of the 
induced subgraphs is perfect.  Consequently, it suffices to consider  plane triangulations
that are both $2$-connected and have no edge separators.  

A triangle $\Delta$  in $G$ consisting of vertices $x,y,z$, is a separating triangle in
$G$ if the interior of $\Delta$  as well as the exterior of 
$\Delta$ contain at least one vertex.  Let $Int(\Delta)$ and 
$Ext(\Delta)$ denote the set of vertices in the interior and exterior of $\Delta$.
It is not hard to see that $G$ is perfect if and only if the subgraphs induced
by $Int(\Delta)\cup \{x,y,z\}$ and $Ext(\Delta)\cup \{x,y,z\}$ are perfect.  
Thus, a separating triangle in $G$ splits $G$ into two induced subgraphs such
that $G$ is perfect if and only if both the induced subgraphs are perfect.  
Consequently, we assume hereafter that $G$ does not contain any separating 
triangles as well.  

A \textit{W-triangulation} is a $2$-connected plane near-triangulation that does not
contain any edge separator or a separating triangle \cite{SalamCWKS19}. It is easy to see that the closed 
neighbourhood $N[x]$ of any internal vertex $x$ in a W-triangulation
$G$ will induce a wheel; for otherwise $G$ will contain a separating triangle. It was shown by \citet{SalamCWKS19}
that a plane W-triangulation
that does not contain any induced wheel on five vertices is not 
perfect if and only if it contains either a vertex or a face, the boundary of the exterior face of the closed 
neighbourhood of which, induces an odd hole.  However, their proof strategy was crucially dependent on the 
graph being free of induced wheels on five vertices.  

Let $X$ be a set of vertices in a graph $G$. The \textit{local neighbourhood} of $X$ is defined as the 
subgraph $G[N[X]]$ induced by $X$ and its neighbours in $G$.  Given a W-triangulation $G$, our objective is to show that if
$G$ is not perfect, then there exists a small connected induced subgraph $X$ in $G$ 
whose local neighbourhood has an odd hole as the boundary of its exterior face. 
We will show that $X$ will either be a vertex, 
an edge or a facial triangle of $G$.      

First, observe that if an internal vertex $x$ of $G$ has odd degree, then the local neighbourhood of $X=\{x\}$ is a wheel, 
whose exterior boundary is an odd hole.  Thus, the non-trivial graphs to consider
are W-triangulations in which all internal vertices are of even degree.  
An \textit{even W-triangulation} is a W-triangulation in which 
every internal vertex has even degree \cite{SalamCWKS19}. 

Figure~\ref{fig00} shows an even W-triangulation $G$
that is not perfect \cite{SalamCWKS19}.  Here, the local neighbourhood of the 
facial triangle consisting of vertices $x,y$ and $z$ has an odd hole as its exterior boundary, and hence
we can choose $X=\{x,y,z\}$.  Note that, for this particular graph,
no smaller substructure (a vertex or an edge) exists, whose local
neighbourhood has an odd hole as its exterior boundary.  

\begin{figure}[!htb]
\centering
 \includegraphics[scale=0.75]{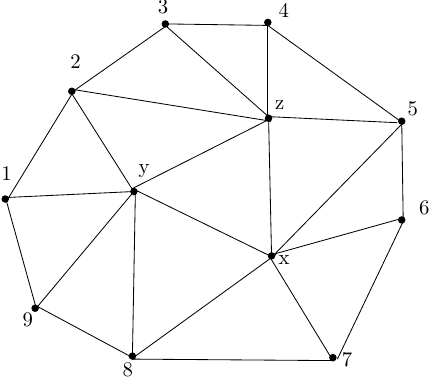}
\caption{The local neighbourhood of the facial triangle $\Delta=\{x,y,z\}$ has an odd hole in its boundary.}
\label{fig00}
\end{figure}

In Section~\ref{sec:perfect}, we prove that every non-perfect even W-triangulation $G$ contains a subset of vertices $X$, 
which consists of either the endpoints of an edge or a facial triangle, such that the induced subgraph $G[N[X]]$ has an odd hole as the boundary 
of its exterior face. This yields:  
 \begin{Theorem} \label{thm:main-thm}
 Let $G=(V,E)$ be a plane near-triangulated graph. $G$ is not perfect if and only if there exists a vertex, 
 an edge or a triangle, the exterior boundary of the local 
neighbourhood of which, is an odd hole.
  \end{Theorem}
In Section~\ref{perfect plane}, we describe an $O(n^2)$ 
algorithm that uses Theorem~\ref{thm:main-thm} to check
whether a plane near-triangulation $G$ of $n$ vertices is perfect. 
\section{Perfect plane near-triangulations}\label{sec:perfect}
Let $G=(V,E)$ be an even non-perfect W-near-triangulated graph.
A \textit{minimal odd hole} $C$ in $G$ is defined as an odd hole 
such that there is no other odd hole in $C \cup Int(C)$. 
Let $C$ be a minimal odd hole in $G$ and let $S=\{a_0,a_1,\ldots,a_s\}$ be the set of vertices in $C$, listed in clockwise order. 
To avoid cumbersome notation, hereafter a reference to a vertex $a_i\in S$ for $i>s$ may be inferred as reference to 
the vertex $a_{i\bmod (s+1)}$. With this notation, we have   
$a_ia_{i+1}\in E(G)$, for $0\leq i \leq s$.  Since $C$ is an odd hole, $s$ must be even and $s\geq 4$. Throughout the paper, the length of a path (respectively, cycle)
will be the number of edges in the path (respectively, cycle).  

Let $H$ be the subgraph of $G$ induced by the vertices in $Int(C)$. We study the structure of $H$ in detail, in order to derive our perfectness characterization. Since $G$ is a plane near-triangulation, $H$ is non-empty.
We show next that $H$ is a $2$-connected plane near-triangulation. 
\begin{Claim}\label{claim:connect}
 $H$ has at least $3$ vertices. Moreover, for any vertex $b \in V(H)$, the neighbours of $b$ on $C$ (if any) are consecutive vertices of $C$. 
\end{Claim}
\begin{proof}
 Since $G$ is an even W-triangulation, $H$ cannot be a single vertex as otherwise the degree of that vertex would be $|S|$, 
which is odd. The number of vertices in $H$ cannot be two, as in that case 
the two vertices must be adjacent, with even degree and having exactly two common neighbours on $C$.  However, this
will contradict the parity of the number of vertices in $C$. Thus, $H$ has at least $3$ vertices. (Note that it is possible
for $H$ to have exactly three vertices as in Figure~\ref{fig00}).  

Let $b$ be a vertex of $H$ with at least one neighbour on $C$. Without loss of generality, we may assume that $a_0$ is adjacent to $b$. 
For contradiction, assume that the neighbours of $b$ on $C$ are not consecutive. That is, for some $0<i<j<k<l\le s+1$, we have
$\{a_0, a_1, \ldots a_{i-1}\}\subseteq N(b)$, $\{a_j, a_{j+1},\ldots, a_{k-1}\} \subseteq N(b)$ and $a_{l \bmod (s+1)} \in N(b)$; but
$\{a_i, a_{i+1}, \ldots, a_{j-1}\}  \cap N(b) =\emptyset$ and $\{a_k, a_{k+1}, \ldots, a_{l-1}\} \cap N(b) =\emptyset$ (see Figure~\ref{fig01}). 
Note that, it is possible to have $a_{l \bmod (s+1)} =a_0$. 

Let $P_1, P_2, P_3, P_4$ and $P_5$ be subpaths of $C$ defined as follows: 
$P_1 = a_0 a_1 \ldots a_{i-1}$, $P_2 = a_{i-1}a_i a_{i+1} \ldots a_{j}$, $P_3 = a_j a_{j+1} \ldots a_{k-1}$, $P_4 = a_{k-1} a_k a_{k+1} 
\ldots a_{l}$ and $P_5 = a_l a_{l+1} \ldots a_{s}a_0$. Note that, the union of these five paths is $C$. 
We know that the closed neighbourhood of $b$ induces a wheel. 
Let $R$ denote the cycle forming the exterior boundary of this wheel. Since $G$ is an even W-triangulation, $R$ is of even length. 
Let $Q_1=x_1 x_2 \ldots x_p$, $Q_2=y_1 y_2 \ldots y_q$ be subpaths of $R$ that are vertex disjoint from $C$, such that 
the concatenation of the paths $P_1, a_{i-1} x_1, Q_1, x_p a_j, P_3, a_{k-1} y_1, Q_2, y_q a_l$ (in that order) forms a subpath of $R$.  

The path $P_2$ has at least two edges, because $i<j$. It is easy to see that $b a_{i-1} P_2 a_j b$ is an induced cycle in $G$ and since $C$ 
is assumed to be a minimal odd hole, $P_2$ must be of even length. 
Similarly, the path $P_4$ also has at least two edges and is of even length.   
Note that, the edge set $(E(C) \setminus E(P_2)) \cup E(Q_1) \cup \{a_{i-1} x_1, x_p a_j\}$ is another induced cycle in $G$. Hence, 
the path $Q_1$ must be of odd length. Similarly, the path $Q_2$ is also of odd length. 
But, this implies that the induced cycle formed by the edge set 
$(E(C) \setminus (E(P_2) \cup E(P_4))) \cup E(Q_1) \cup E(Q_2) \cup \{a_{i-1} x_1, x_p a_j, a_{k-1} y_1, y_q a_l\}$ is an odd hole.
This contradicts the minimality of $C$ and hence we can conclude that the neighbours of $b$ on $C$ must be consecutive.
\end{proof}
\begin{figure}[!htb]
\centering
 \includegraphics[scale=0.70]{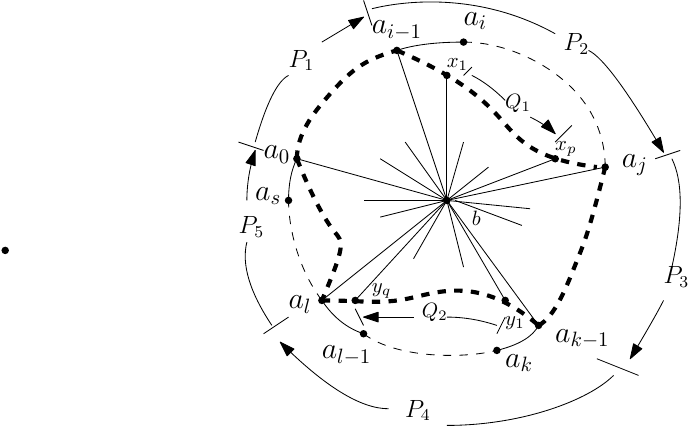}
\caption{b is a vertex with non consecutive neighbours in $C$}
\label{fig01}
\end{figure}

\begin{lemma}\label{lem:H-connected}
 $H$ is a $2$-connected near-triangulation with at least three vertices. 
\end{lemma}
\begin{proof}
 Since $G$ is a plane near-triangulation, it is easy to see that $H$ is a plane near-triangulation. Hence, 
 by Claim~\ref{claim:connect}, it only remains to prove that $H$ has no cut vertices.
 
 For contradiction, suppose $b$ is a cut vertex in $H$. Since $H$ is a near-triangulation, $b$ must be a vertex on the exterior face of $H$. (Any internal vertex $v$ of $H$ has the property that each pair of its neighbours lie on a cycle. Thus $v$ cannot be a cut vertex of $H$.) Moreover, as $H$ is the subgraph induced by the vertices in $Int(C)$ and as $G$ is a near-triangulation, every vertex on the exterior boundary of $H$ must be adjacent to some vertex in $C$. Hence $b$ must have at least one neighbour on $C$
 By Claim~\ref{claim:connect}, neighbours of $b$ on $C$ are consecutive. 
 Without loss of generality, let $N(b) \cap V(C)=\{a_0, a_1, \ldots, a_i\}$, for some $0 \le i < s$. 
 Since $G$ is a W-triangulation, the neighbourhood of $b$ induces a wheel
 with $b$ as the center. Let $R$ be the induced cycle formed by the neighbours of $b$ in $G$.  
 Since $G$ has no separating triangles, $a_0, a_1, \ldots, a_i$ should be consecutive vertices in $R$ as well. Let $H_1$ and $H_2$ be two connected 
 components of $H \setminus b$. Let $v_1$ be a neighbour of $b$ in $H_1$ and $v_2$ be a neighbour of $b$ in $H_2$. Both $v_1$ and $v_2$ are 
 vertices that belong to the cycle $R$ and they cannot be consecutive on $R$. 
 Let $P_1$ and $P_2$ denote the two edge disjoint paths between $v_1$ and $v_2$ in $R$ such that their union is $R$. 
 As $v_1$ and $v_2$ are not connected in $H\setminus b$,  both $P_1$ and $P_2$ must intersect $C$. 
 These intersections happen on vertices in $V(C) \cap V(R) \subseteq N(b)$.  
 Hence there exist $0 \le j, k\le i$ such that $a_j \in V(P_1) \cap V(C)$ and $a_k \in  V(P_2) \cap V(C)$.  
 Note that, if we delete $v_1$ and $v_2$ from $R$, $a_j$ and $a_k$ get disconnected from each other. 
 However, this is impossible, since vertices $N(b) \cap V(C)=\{a_0, a_1, \ldots, a_i\}$ are known to be consecutive on $R$. 
 Hence, $H$ is $2$-connected.
\end{proof}
Consequently from Lemma~\ref{lem:H-connected} we have:  
\begin{Corollary}\label{cor:H}
The boundary of the exterior face of $H$ has at least three vertices.
 \end{Corollary}

%
%
Let $C'$ be the cycle forming the boundary of the exterior face of $H$. If $C'$ is a triangle, then it must be a facial triangle in $G$, as $G$ is assumed
to contain no separating triangles. In this case, the exterior face of the closed 
neighbourhood of $C'$ is the odd hole $C$, and Theorem~\ref{thm:main-thm} is immediate.  Hence, we assume hereafter that
$C'$ is not a triangle.  

Let $T=\{b_0,b_1,\ldots,b_t\}$, $t\geq 3$   
be the vertices of $C'$, listed in clockwise order. 
To simplify the notation, reference to a vertex  
$b_j\in T$ for $j>t$ may be inferred as reference to 
the vertex $b_{j\bmod (t+1)}$.  
It is easy to see that every vertex in $T$ is a neighbour of at least one vertex in $S$ and vice versa. 
Let $G'=(V',E')$ be the subgraph of $G$ with $V'=S\cup T$ and 
$E'=\{uv: u\in S , v\in T\} \cup E(C) \cup E(C')$. Note that, $b_ib_{i+1}\in E'$, for all $0\leq i \leq t$; 
but $E'$ excludes chords in $C'$.  Since $G$ is a plane near-triangulation, the following observation is immediate.
\begin{Observation}\label{obs:neighbour}
For every $0\leq i\leq s$ and $0 \leq j \leq t$:
\begin{enumerate}
\item The neighbours of $a_i$ in $C'$ must be $b_l,b_{l+1}\ldots b_m$ 
for some consecutive integers $l,l+1\ldots m$. 
\item The neighbours of $b_j$ in $C$ must be $a_k,a_{k+1}\ldots a_r$ 
for some consecutive integers $k,k+1\ldots r$.
\item $a_{i}$ and $a_{i+1}$ must have a common neighbour in $C'$.
\item $b_{j}$ and $b_{j+1}$ must have a common neighbour in $C$.
\item The minimum degree of any vertex of $G'$ is at least $3$.
\end{enumerate}  
\end{Observation}
Consider a vertex $a_i\in S$.  Suppose $b_l,b_{l+1},\ldots b_{m}$ are the neighbours of $a_i$ in $C'$.  
If there is an edge between some non-consecutive vertices $b_p$ and $b_q$ for some $l \le p, q \le m$, 
then $a_i,b_p$ and $b_q$ will form a separating triangle in $G$, a contradiction.  Hence we have: 
\begin{Observation}\label{Obs:triangle}
For any $a_i \in S$, there exists an edge between two neighbours of $a_i$ in $C'$ if and only if they are consecutive in $C'$.   
\end{Observation}


\begin{lemma}\label{evenC'}
For any vertex $b_j\in T$, $deg_{G'}(b_j)$ is either $3$ or an even number greater than $3$. 
\end{lemma}
\begin{proof}
Let $deg_{G'}(b_j)\neq 3$ and be odd. Since $deg_{C'}(b_j)=2$, the neighbours of $b_j$ in $C$ form a path (say $Z$) of even length ($\geq 2$). 
Let the neighbours of $b_j$ in $C$ be $a_i,a_{i+1},\ldots,a_k$ in the clockwise order (see Figure~\ref{fig0}(a)). 
Note that $a_k a_i \notin E(G)$. 
Further, since $C$ is an odd hole, the length of the subpath $Y$ of $C$ from $a_k$ to $a_i$ in clockwise direction, 
must be an odd number $\geq 3$. 
Hence the cycle formed by replacing the path $Z$ in $C$ with the edges $a_ib_j,b_ja_k$ will be an odd hole, distinct from $C$, 
with vertices chosen only from $S$ and $Int(C)$. This contradicts the minimality of $C$.

\end{proof}
\begin{figure}[!htb]
\centering
 \includegraphics[scale=0.50]{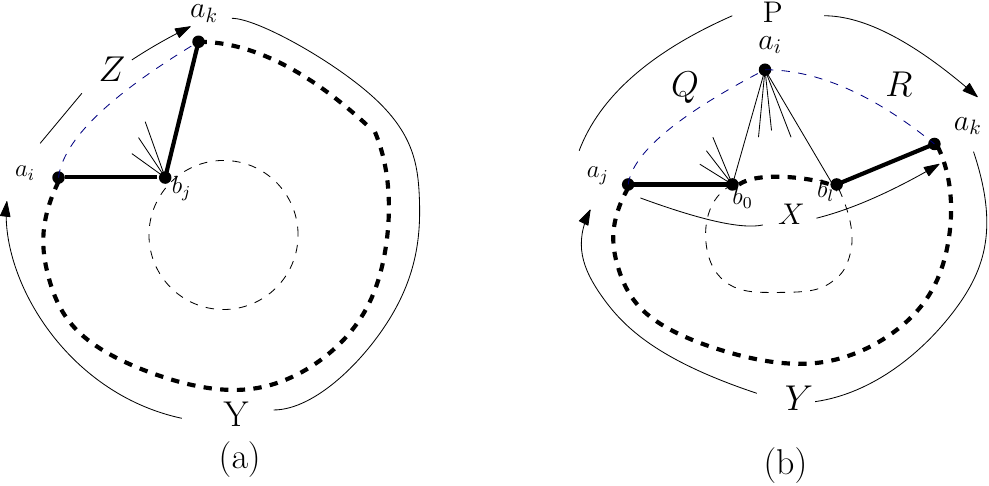}
\caption{ (a) The degree of vertex $b_i$ is odd. (b) The degree of vertex $a_i$ is odd}
\label{fig0}
\end{figure}

\begin{lemma}\label{evenC} 
For any vertex $a_i\in S$, $deg_{G'}(a_i)$ is either $3$ or an even number greater than $3$. 
\end{lemma}
\begin{proof}
Suppose $deg_{G'}(a_i) > 3$. With no loss of generality, let the neighbours of $a_i$ in $C'$ be  
$b_0,b_1,\ldots,b_l$, $l\geq 1$ in the clockwise order (see Figure~\ref{fig0}(b)). For contradiction, suppose $l$ is even. Then, the length of the path
$b_0,b_1\ldots, b_l$ is even.  
Let the neighbours of $b_0$ in the clockwise direction in $C$  be $a_{j},a_{j+1},\ldots,a_i$ and the neighbours of $b_l$ in the clockwise direction in $C$ be 
$a_i,a_{i+1},\ldots,a_{k}$. Vertices $b_0$ and $b_l$ cannot be adjacent in $G'$, as otherwise $a_i b_0 b_l$ would be a separating triangle in $G$ which is a contradiction. Hence, $b_l \ne b_t$ and $b_{l+1}\ne b_0$. Now, by Observation~2.3, $b_l$ and $b_{l+1}$ have a common neighbor in $C$. This common neighbor is different from $a_i$, since $b_{l+1}$ is not a neighbour of $a_i$. Thus, $deg_{G'}(b_l)>3$.   Similarly, $deg_{G'}(b_0)>3$.   
 Hence, we may assume without loss
of generality that $j<i$.  

Let $P$ be the path in $C$ from $a_{j}$ to $a_{k}$ through $a_i$. As $deg_{G'}(b_0)$  and $deg_{G'}(b_l)$ are even (by Lemma~\ref{evenC'}), the length of path $P$ is even, $\geq 2$. Let $X$ be the path $a_jb_0b_1\ldots b_la_k$. Since $l$ is assumed to be even, $deg_{G'}(a_i)$ is odd, and hence $X$ must be of even length. Note that, $a_k \neq a_j$, as in that case the path $P$ will be the whole (odd) cycle $C$ which is impossible as $P$ has even length. Thus, the edges $a_jb_l$ and $a_kb_0$ cannot be present in $G'$ as otherwise we would have $a_k=a_j$. Hence, unless 
$a_ka_j$ is an edge in $G'$, then $X$ is an induced path in $G$.  
However, $a_ka_j$ cannot be an edge in $G'$ as otherwise $X$ along with
the edge $a_ka_j$ induces an odd hole, distinct from $C$, with
vertices chosen only from $S$ and $Int(C)$, which is impossible. Thus, we conclude that $X$ is an induced path in $G'$. 
Consequently, the cycle formed by replacing  even length path $P$ in $C$ with the even length induced path $X$  will be an odd hole, distinct from $C$, with
vertices chosen only from $S$ and $Int(C)$. This contradicts the choice of $C$.
\end{proof}

We introduce some notation.  
A vertex $a_i\in S$ with $deg_{G'}(a_i)=4$ (respectively $deg_{G'}(a_i)>4$) will be called an $\alpha$ vertex (respectively $\alpha'$ vertex).  
Similarly a vertex $b_j\in T$ with $deg_{G'}(b_j)=4$ (respectively $deg_{G'}(b_j)>4$) will be called an $\beta$ vertex (respectively $\beta'$ vertex).  
Every degree three vertex in $S$ (respectively $T$) will be called a $\gamma$ vertex 
(respectively $\delta$ vertex). 
Let $V_\alpha,V_{\alpha'},V_\beta, V_{\beta'},V_{\gamma}$ and $V_{\delta}$ denote the set of 
$\alpha, \alpha',\beta,\beta',\gamma$ and $\delta$ vertices respectively and $N_\alpha,N_{\alpha'},N_\beta, N_{\beta'},N_{\gamma}$ and $N_{\delta}$ 
denote the number of $\alpha, \alpha',\beta,\beta',\gamma$ and $\delta$ vertices respectively (See Figure~\ref{fig1}). 

\begin{figure}[h]
\centering
 \includegraphics[scale=0.65]{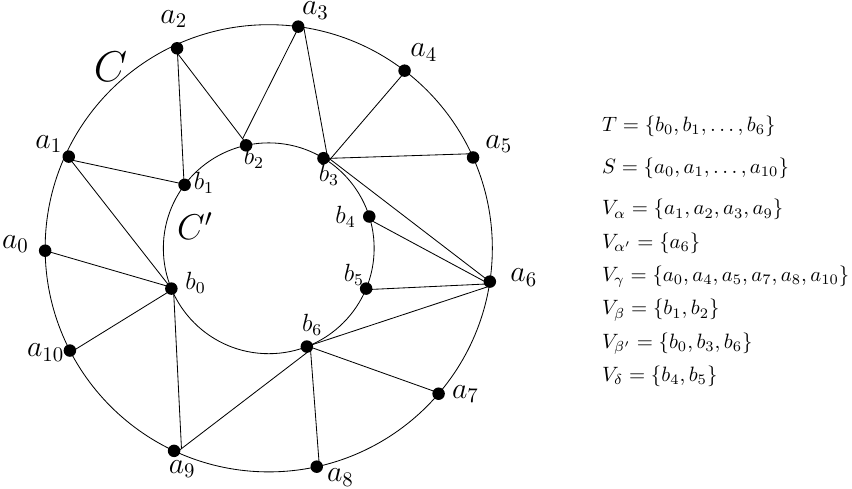}
\caption{An illustration of the notation.}
\label{fig1}
\end{figure}

With the above notation, the following observation is immediate. 
\begin{Observation}\label{obs:gamma}
~
\begin{enumerate}
\item Every $\alpha$ and $\alpha'$ vertex (respectively $\beta$ and $\beta'$ vertex) has exactly two neighbours in $T$ (respectively $S$) 
that are not $\delta$ vertices (respectively $\gamma$ vertices).
\item Every $\gamma$ vertex (respectively $\delta$ vertex) must have exactly one neighbour in $T$ (respectively $S$). 
Moreover the neighbour must be a $\beta'$ vertex (respectively $\alpha'$ vertex).
\end{enumerate}  
\end{Observation}
 \begin{lemma}\label{lem:numgamma}
$N_\gamma$ and $N_\delta$ are even.  
\end{lemma}
\begin{proof}
  Let $G''$ be the bipartite subgraph of $G'$ with vertex set $V''=V_\delta \cup V_{\alpha'}$ and 
 edge set $E''=\{a_i b_j \colon a_i \in V_{\alpha'}, b_j \in V_\delta\}$.  
  Consider any vertex $a_i \in V_{\alpha'}$. We know that  in $G'$, $a_i$ has exactly two neighbours in $S$. 
 Since by Lemma~\ref{evenC}, $deg_{G'}(a_i)$ is even, it follows that  in $G'$, $a_i$ has an even number of neighbours from $T$. 
 From Observation~\ref{obs:gamma}, $a_i$ has exactly two neighbours in $T$ that are not in $V_\delta$. Hence, the number of 
 edges from $a_i$ to $V_\delta$ in $G'$ must be even. 
 Therefore, in the bipartite graph $G''$, every vertex in $V_{\alpha'}$ has an even degree. 
 By Observation~\ref{obs:gamma}, $deg_{G''}(b_j)=1$, for each $b_j \in  V_\delta$. 
 Consequently, $N_\delta = |V_\delta|= \sum_{a_i \in V_{\alpha'}}{deg_{G''}(a_i)}$ 
 is an even number. 
 
 The proof for the claim that $N_\gamma$ is even is similar.  
\end{proof}

Note that $|S|=N_{\gamma}+N_\alpha+ N_{\alpha'}$. 
By Lemma~\ref{lem:numgamma}, $N_{\gamma}$ is even. As $C$ is an odd hole, we have:  

\begin{Corollary}\label{cor:numalpha}
$N_\alpha+ N_{\alpha'}$ is odd. 
\end{Corollary}
\begin{lemma}\label{obs:Nalphabeta}
~
\begin{enumerate}
\item Every vertex in $S$ is a neighbour of a $\beta$ or a  $\beta'$ vertex.
\item Every vertex in $T$ is a neighbour of a $\alpha$ or a  $\alpha'$ vertex.
\item $N_{\alpha}+N_{\alpha'}=N_{\beta}+N_{\beta'}$. 
\end{enumerate}
\end{lemma}
\begin{proof}
The first two parts are easy to see. 
To prove the third part, consider the bipartite subgraph $G''=(V'', E'')$ of $G'$ with vertex set 
$V''=V_{\alpha} \cup V_{\alpha'}\cup V_{\beta}\cup V_{\beta'}$ and edge set 
$E''=\{a_i b_j \colon a_i \in V_{\alpha} \cup V_{\alpha'}, b_j \in V_{\beta}\cup V_{\beta'} \}$.
By Observation~\ref{obs:gamma}, every $\alpha$ and $\alpha'$ vertex has exactly two neighbours in $T$ that are either $\beta$ vertices 
or $\beta'$ vertices. 
Similarly, each $\beta$ and $\beta'$ vertex have exactly two neighbours in $S$ which are either $\alpha$ 
or $\alpha'$ vertices. Thus, $G''$ is a $2$-regular bipartite graph. Hence,  $N_{\alpha}+N_{\alpha'}=N_{\beta}+N_{\beta'}$.
\end{proof}
Combining Lemma~\ref{lem:numgamma}, Corollary~\ref{cor:numalpha} and Lemma~\ref{obs:Nalphabeta}, 
we see that $|T|=N_\delta+N_{\beta}+N_{\beta'}$ is odd. This means that $C'$ is an odd cycle. As $C'$ cannot be an odd hole and we 
have assumed that it is not a triangle, we have:
\begin{Observation}\label{obs:chord}
 $C'$ is an odd cycle with at least one chord in $G$ connecting vertices in $C'$. 
\end{Observation}

\begin{figure}[!htb]
\centering
 \includegraphics[scale=0.50]{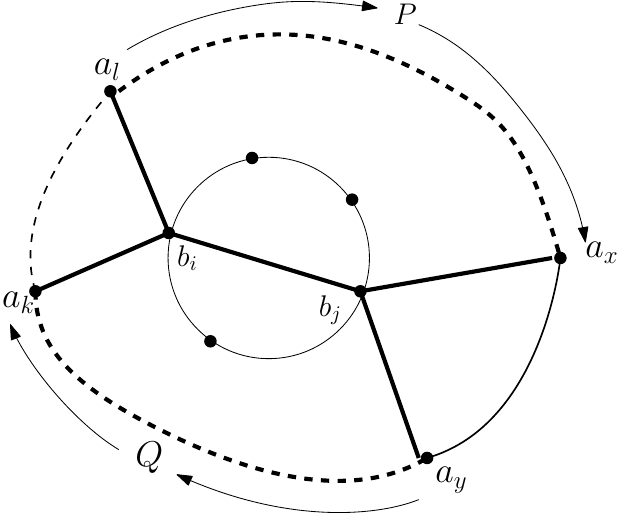}
\caption{$b_i \in V_{\beta'}$ and $b_j \in V_\beta$.  }
\label{fig2}
\end{figure}
The following lemma shows that any chord in $C'$ must have a $\delta$ vertex as one of its end points.
\begin{lemma}\label{lem:multiC}
If $b_ib_j$ is a chord in $C'$ then $\{b_i,b_j\} \cap V_\delta \ne \emptyset$.
\end{lemma}
\begin{proof}
Suppose $\{b_i,b_j\} \cap V_\delta =\emptyset$ and $i < j$. 
Let $a_k,a_{k+1},\ldots,a_l$ be the neighbours of $b_i$ in $C$ and let $a_x,a_{x+1},\ldots,a_y$ be the neighbours of $b_j$ in $C$,  
both taken in clockwise order (see Figure~\ref{fig2}). As $b_ib_j$ is a chord, there exists at least one vertex between $b_i$ and $b_j$ in $C'$ 
(in both directions). Hence $a_l \neq a_x$ (otherwise the vertices $a_x,b_i,b_j$ will form a separating triangle). 
Similarly, $a_k \neq a_y$. As $b_i$ and $b_j$ have even number of neighbours in $C$, it is easy to see that either 
the clockwise path $P$ from $a_l$ to $a_x$ or the clockwise path $Q$ from $a_y$ to $a_k$ in $C$ must be even, 
for otherwise $C$ cannot be an odd hole. With no loss of generality, assume that the length of $P$ is even. 
Then the vertices in the path $P$ along with the edges $\{a_xb_j,b_jb_i,b_ia_l\}$ forms an odd hole consisting of vertices chosen only 
from $S$ and $Int(C)$, a contradiction.
\end{proof}
By Lemma~\ref{lem:multiC}, $C'$ contains at least one chord connecting a vertex 
$b_j\in V_\delta$ to some other vertex in $C'$. By Observation~\ref{obs:gamma}, the neighbour of $b_j$ in $C$ must be an $\alpha'$ vertex.  Hence,  
 \begin{Corollary}\label{cor:nalpha1}
There exists at least one $\alpha'$ vertex in $C$. 
\end{Corollary}
The next lemma shows that if a chord in $C'$ connects 
two $\delta$ vertices, then their neighbours in $C$ will be adjacent.
\begin{lemma}\label{lem:gammachord}
Let $\{b_i,b_j\}\subseteq V_\delta$. Let $a_x$ and $a_y$ be the neighbours of $b_i$ and $b_j$ respectively in $C$. 
If $b_ib_j$ is a chord in $C'$ (in $G$), then $a_xa_y$ is an edge in $C$. 
\end{lemma}
\begin{proof}
Let $\{b_i,b_j\} \subseteq V_\delta$ and suppose $b_ib_j$ is a chord in $G$ connecting vertices of $C'$. 
Let $a_x$ and $a_y$ be the unique neighbours of $b_i$ and $b_j$ respectively in $C$ (Observation~\ref{obs:gamma}).  
It follows that $a_x\neq a_y$, as otherwise the vertices $b_i,a_x,b_j$ will form a separating triangle in $G$ (See Figure~\ref{fig3}(a)). 
Since $C$ is an odd hole, exactly one of the two paths from $a_x$ to $a_y$ through vertices in $C$ must be of even length. 
With no loss of generality, assume that the path $P$ from $a_y$ to $a_x$ in the clockwise direction is even (see Figure~\ref{fig3}(b)). 
Then if $a_xa_y$ is not an edge in $C$, the edges in the path $P$ together with the edges $a_xb_i,b_ib_j$ and $b_ja_y$ forms an odd hole consisting 
of vertices chosen only from $S$ and $Int(C)$ (see Figure~\ref{fig3}(b), (c)), a contradiction.
\end{proof}

 \begin{figure}[!htb]
\centering
 \includegraphics[scale=0.60]{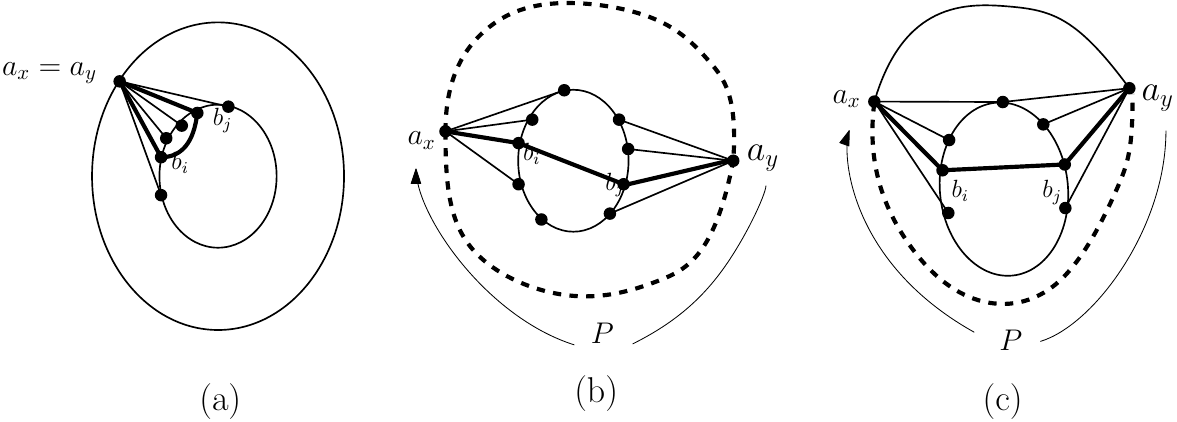}
\caption{(a) $a_x=a_y$ (b) $a_xa_y$ is not an edge in $G$ (c) $a_xa_y$ is an edge}
\label{fig3}
\end{figure}

Next we bound the number of $\alpha$ and $\alpha'$ vertices in $C$. A bound on the number of $\beta$ and $\beta'$ vertices follows from this.
  \begin{lemma}\label{nalphabeta3}
 $N_\alpha+ N_{\alpha'}\leq 3$
 \end{lemma}
 \begin{proof}
  Suppose $N_\alpha+ N_{\alpha'}> 3$. As $N_\alpha+ N_{\alpha'}$ is odd (by Corollary~\ref{cor:numalpha}), 
 $N_\alpha+ N_{\alpha'}\geq 5$. 
 By Corollary~\ref{cor:nalpha1}, an $\alpha'$ vertex $a_i$ must exist in $C$. 
 By the definition of an $\alpha'$ vertex and by Lemma~\ref{evenC}, we know that $a_i$ has an an even number ($\ge 6$) of neighbours in $G'$, of which
 exactly two are on $C$.
 Without loss of generality, let $b_q,b_{q+1},\ldots,b_r$ ($r\geq q+3$) 
 be the neighbours of $a_i$ in  $T$ (see Figure~\ref{fig4}(a)). Let the neighbours of $b_q$ and $b_r$ in $C$ in clockwise 
 order be $a_x,a_{x+1}\ldots, a_i$ and $a_i,a_{i+1},\ldots,a_y$ respectively. 
 Clearly, $a_x$ and $a_y$ are in $V_\alpha \cup V_{\alpha'}$. If $x=y$, then $V_\alpha \cup V_{\alpha'}=\{a_x, a_i\}$, contradicting our assumption
 that $N_\alpha+ N_{\alpha'} \ge 5$. Further, since $a_x, a_i$ and $a_y$ are the only vertices in $V_\alpha \cup V_{\alpha'}$ in the clockwise subpath of 
 $C$ from $a_x$ to $a_y$, at least two internal vertices of the clockwise subpath of $C$ from $a_y$ to $a_x$ must belong to 
 $V_\alpha \cup V_{\alpha'}$, since $N_\alpha+ N_{\alpha'} \ge 5$.

 Let $b_p,b_{p+1},\ldots, b_{q}$ and 
 $b_r,b_{r+1},\ldots,b_{s}$ be the neighbours of $a_x$ and $a_y$ in $C'$ respectively in clockwise order. 
 Note that $b_p$ and $b_s$ are distinct, non-adjacent vertices in $C'$.  This is because $N_{\beta}+N_{\beta'}=N_{\alpha}+N_{\alpha'}\geq 5$
 (part (3) of Lemma~\ref{obs:Nalphabeta}).  
 Let the neighbours 
 of $b_p$ and $b_s$ in $C$ be $a_g,a_{g+1},\ldots,a_x$ and $a_y,a_{y+1},\ldots,a_h$ respectively. 
 As $N_\alpha+ N_{\alpha'}\geq 5$, $a_g \neq a_h$. Since $b_p,b_q,b_r$ and $b_s$ are elements of $V_\beta \cup V_{\beta'}$, 
 their degrees must be even (by Lemma~\ref{evenC'}). Hence the path from $a_g$ to $a_h$ through $a_x,a_i$ and $a_y$ along 
 vertices in $C$ is of even length. Consequently, as $C$ is an odd hole, the (chordless) path $L$ from $a_h$ to $a_g$ in clockwise order 
 in $C$ (as shown in Figure~\ref{fig4}(a)) must be of odd length ($\geq 1$). We will show that there exists another odd hole in $G$, 
 consisting of vertices chosen only from $S$ and $Int(C)$, contradicting the choice of $C$. 
 
Consider the path $P=a_gb_pb_{p+1}\ldots b_{q}a_ib_rb_{r+1}\ldots b_{s}a_h$. 
Since $a_x$ and $a_y$ are of even degree (by Lemma~\ref{evenC}), $P$ is of even length (see Figure~\ref{fig4}(a)). This path, together with the
with the path $L$, form an odd cycle of length at least seven. If we prove that this cycle is chordless, it will be an odd hole, 
consisting of vertices chosen only from $S$ and $Int(C)$, a contradiction to the choice of $C$ as desired.  

Hence we analyze the possible chords in $G$ for the cycle formed by edges of $P$ and the edges of $L$.  
By Observation~\ref{Obs:triangle}, there cannot be any chord in $P$ connecting any two vertices in the set $\{b_p,b_{p+1},\ldots b_q\}$ or any two vertices in the set $\{b_r,b_{r+1},\ldots,b_s\}$. For the same reason, a $b_qb_r$ chord also cannot exist. 
Moreover, since as $b_q,b_p, b_r$ and $b_s$ are elements in $V(\beta)\cup V(\beta')$, by Lemma~\ref{lem:multiC}, there will not be a chord between them.
Further, as there is no edge connecting $a_x$ and $a_y$, no chord exists between a vertex in the set $\{b_{p+1},\ldots b_{q-1}\}$ and a vertex in the set $\{b_{r+1}\ldots b_{s-1}\}$ (by Lemma~\ref{lem:gammachord}).  
Hence, chords in $P$ can exist only between a vertex in the set $\{b_p,b_q\}$ and a vertex in the set $\{b_{r+1},\ldots,b_{s-1}\}$ 
or between a vertex in the set $\{b_r,b_s\}$ and a vertex in the set $\{b_{p+1},\ldots ,b_{q-1}\}$. 
We systematically rule out these possibilities below.
\begin{itemize}
\item \textbf{Case 1} There exists a $b_pb_k$ chord in $C'$ for some 
$b_k \in \{b_{r+1}\ldots,b_{s-1}\}$ (Figure~\ref{fig4}(b)): Let $L'$ be the (chordless) path from $a_y$ to $a_g$ in 
the clockwise direction in $C$. 
Since $ b_s$ is of even degree (by Lemma~\ref{evenC'}), and $L'$ is obtained by combining the path from $a_y$ to $a_h$ along $C$ with the odd 
length path $L$, path $L'$ should be of even length.
Consequently, the path $L'$ along with the edges $a_gb_p,b_pb_k$ and $b_ka_y$ will induce an odd hole consisting of vertices chosen only from $S$ 
and $Int(C)$, a contradiction.

\item \textbf{Case 2} There exists a $b_sb_k$ chord in $C'$ for some 
$b_k \in \{b_{p+1}\ldots,b_{q-1}\}$: This case is symmetric to Case 1.
\item \textbf{Case 3} There exists a $b_qb_k$ chord in $C'$ for some 
$b_k \in \{b_{r+1}\ldots,b_{s-1}\}$ (Figure~\ref{fig4}(c)): Let the neighbours of $a_h$ in $C'$ be 
$b_s,b_{s+1},\ldots,b_m$ and the neighbours of $b_m$ on $C$ be  
$a_h,a_{h+1},\ldots,a_j$ . Note that since $N_\alpha + N_{\alpha'} \geq 5$, by part (3) of Lemma~\ref{obs:Nalphabeta}, $V_{\beta}\cup V_{\beta'}$ should 
contain at least five vertices. Hence we see that $b_m \neq b_p$. Consider the path $Q=a_xb_qb_{q+1}\ldots b_ra_yb_s\ldots b_ma_j$. 
As $a_i$ and $a_h$ are of even degree (by Lemma~\ref{evenC}), the path $Q$ is of even length. 
As $b_q,b_r,b_s$ and $b_m$ are elements in $V(\beta)\cup V(\beta')$, they have even number of neighbours in $C$ (by Lemma~\ref{evenC'}). 
Hence the path from $a_x$ to $a_j$ in clockwise direction in $C$ is of even length. Consequently, as $C$ is an odd hole, the path  $M$ 
from $a_j$ to $a_x$ in $C$ in clockwise order is of odd length. Suppose that the path $Q$ is chordless. 
Then combining the path $Q$ with $M$ yields an odd hole (see Figure~\ref{fig4}(c))  consisting of vertices chosen only from $S$ and $Int(C)$, 
a contradiction. Note that this is true even if  $a_j=a_g$ (i.e., $N_\alpha+ N_{\alpha'}=5$).

Thus it suffices to prove that the path $Q$ is chordless. We rule out each of the following possible cases of chords appearing in $Q$.
\begin{itemize}
\item[(a)] There exists a chord connecting a vertex in the set $\{b_{q+1},\ldots,b_r\}$ and 
a vertex in the set $\{b_s,\ldots,b_m\}$ in $C'$:  
As $b_qb_{k}$ is a chord in $C'$, this is impossible, as otherwise $G$ cannot be planar.
\item[(b)] There exists a $b_qb_m$ chord or a $b_qb_s$ chord: This possibility is ruled out by Lemma~\ref{lem:multiC}.
\item[(c)] There exists a chord $b_q b_l$ such that  $b_l \in \{b_{s+1},\ldots,b_{m-1}\}$ (Figure~\ref{fig4}(d)): 
Let $b_qb_l$ be a chord in $C'$ . Let $M'$ be the path from $a_h$ to $a_x$ in clockwise direction in $C$. 
Since $b_q,b_r$ and $b_s$ are elements in $V(\beta)\cup V(\beta')$, they have even number of neighbours in $C$ (by Lemma~\ref{evenC'}). 
Hence the path from $a_x$ to $a_h$ in clockwise direction in $C$ is of odd length. Consequently, as $C$ is an odd hole, the path $M'$ 
must have even length ($\geq 2$). Hence the path $M'$ along with the edges $a_xb_q,b_qb_l$ and $b_la_h$ will induce an odd hole consisting 
of vertices chosen only from $S$ and $Int(C)$, a contradiction. 
\end{itemize}
Thus we conclude that the path $Q$ is chordless, as required.
\item \textbf{Case 4} There exists a $b_rb_k$ chord in $C'$ for some 
$b_k \in \{b_{p+1}\ldots,b_{q-1}\}$: This case is symmetric to the Case 3.
\end{itemize} 
Hence we conclude that $N_\alpha+ N_{\alpha'}\leq 3$.
 \end{proof}
 
 \begin{figure}[!htb]
\centering
 \includegraphics[scale=0.750]{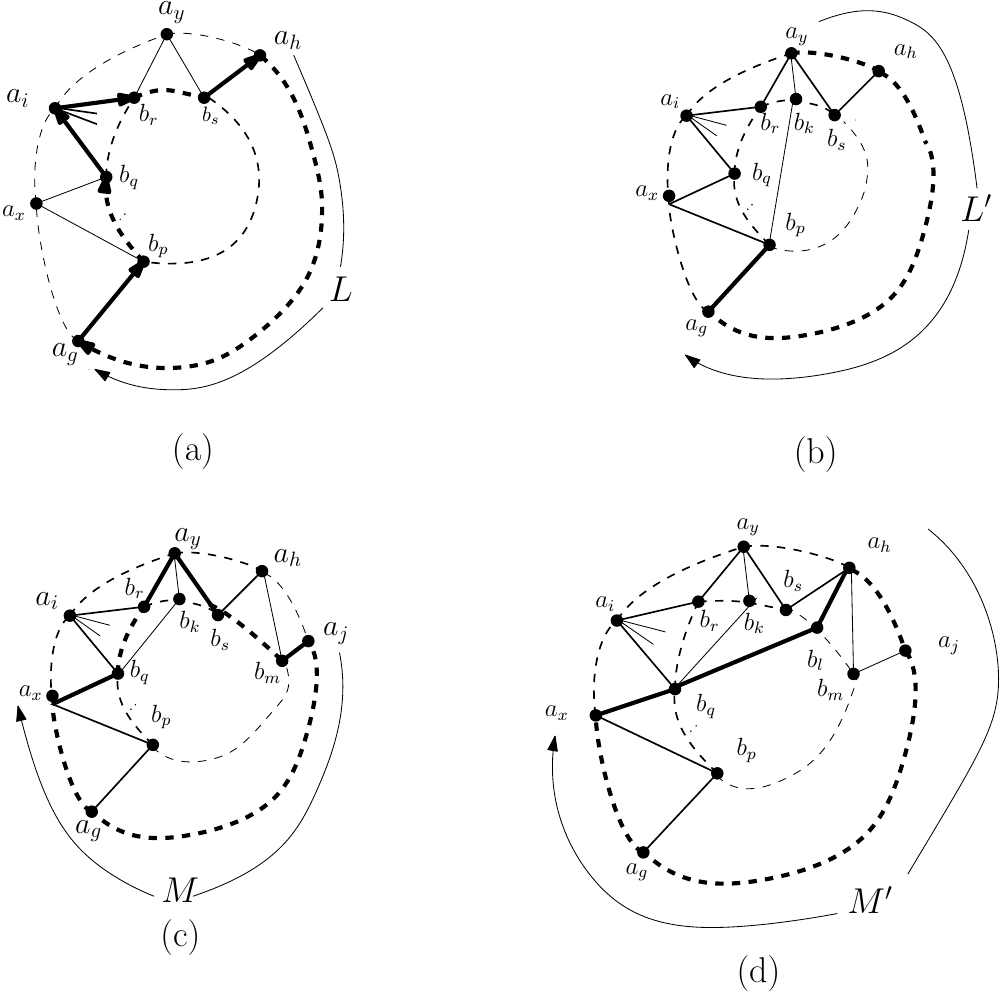}
\caption{$N_\alpha+ N_{\alpha'}\ge 5$.}
\label{fig4}
\end{figure}

By part (3) of Lemma~\ref{obs:Nalphabeta}, We have:
 \begin{Corollary}\label{nbeta3}
$N_\beta+ N_{\beta'}\le 3$.
 \end{Corollary}
Since $N_\alpha+ N_{\alpha'}\leq 3$ and $C$ is an odd hole, $C$ contains at least one $\gamma$ vertex which must be adjacent 
to a $\beta'$ vertex in $C'$ (by part 2 of Observation~\ref{obs:gamma}). Consequently we have:
 \begin{Corollary}\label{nbeta1}
 There exists at least one $\beta'$ vertex in $C'$. That is, $N_{\beta'} \ge 1$.
 \end{Corollary}
    \begin{lemma}\label{lem:alphabeta1}
Let $a_j$ be an $\alpha'$ vertex in $C$. Let $b_x,b_{x+1},\ldots,b_y$ be the neighbours of $a_j$ in $C'$ in clockwise order. 
Then at least one among $b_x$ and $b_y$ must be a $\beta$ vertex. 
  \end{lemma}
  \begin{proof}
  Suppose that $b_x$ and $b_y$ are $\beta'$ vertices.  
  Let $a_i,a_{i+1}\ldots,a_j$ and $a_j,a_{j+1}\ldots, a_k$ ($j\geq i+3$, $k\geq j+3$) be the neighbours of $b_x$ and $b_y$ respectively in $C$, 
  considered in clockwise order (see Figure~\ref{fig5}(a)). As $b_x$ and $b_y$ have even number ($\geq 4$) of neighbours in $C$ (by Lemma~\ref{evenC'}), 
  the path $P$ from $a_i$ to $a_k$ through $a_j$ in $C$ is chordless and has even length ($\geq 6$). Since $C$ is an odd hole, we see that $a_i \neq a_k$. 
  By Lemma~\ref{lem:multiC}, there is no chord connecting $b_x$ and $b_y$ in $G$. 
  Replacing the path $P$ in $C$  by the edges $a_ib_x,b_xa_j,a_jb_y$ and $b_ya_k$ induces an odd hole in $G$ consisting of vertices chosen only from $S$ 
  and $Int(C)$, a contradiction. Hence at least one among $b_x$ and $b_y$ must be a $\beta$ vertex.
  \end{proof}
 
 \begin{figure}[!htb]
\centering
 \includegraphics[scale=0.75]{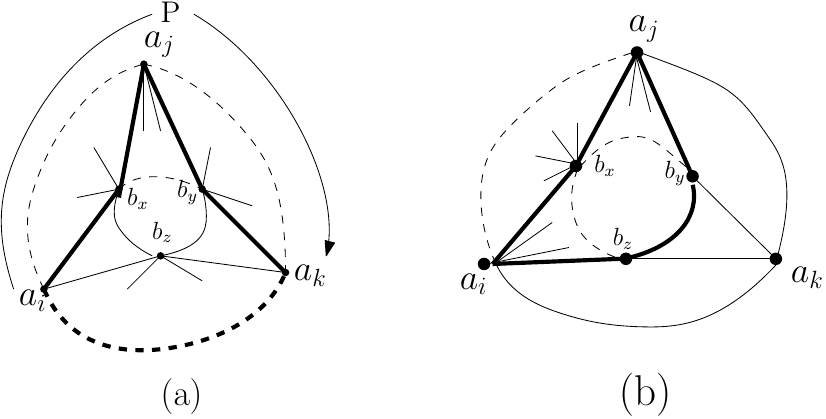}
\caption{(a) $N_{\beta'}=3$ and $N_\delta\neq 0$  (b) $b_x$ is a $\beta'$ vertex and both $a_i$ and $a_j$ are $\alpha'$ vertices. }
\label{fig5}
\end{figure}

Since there exists at least one $\alpha'$ vertex in $C$ (by Corollary~\ref{cor:nalpha1}), we have:
\begin{Corollary} \label{cor:nbeta1}
$N_\beta \ge 1$.
\end{Corollary}
As $N_{\beta}\geq 1$ and $N_{\beta'} \geq 1$ (by Corollary~\ref{cor:nbeta1} and Corollary~\ref{nbeta3}), we have $N_{\beta}+N_{\beta'}\geq 2$. 
Moreover, by part (3) of Lemma~\ref{obs:Nalphabeta} and 
Corollary~\ref{cor:numalpha}, $N_{\beta}+N_{\beta'} (= N_{\alpha}+N_{\alpha'})$ is odd and by Corollary~\ref{nbeta3}, 
$N_{\beta}+N_{\beta'}\leq 3$. Thus we have: 
\begin{Observation}\label{Obs:nbeta2}
  $N_{\alpha}+N_{\alpha'}=N_{\beta}+N_{\beta'}=3$.
  \end{Observation}
 \begin{lemma}\label{lem:betaalpha1}
Let $b_x$ be a $\beta'$ vertex in $C'$ and $a_i,a_{i+1},\ldots,a_j$ be the neighbours of $b_x$ in $C$ in clockwise order. 
If $N_{\alpha}\neq 0$ then at least one among $a_i$ and $a_j$ must be an $\alpha$ vertex. 
\end{lemma}
\begin{proof}
For the sake of contradiction assume that $a_i$ and $a_j$ are $\alpha'$ vertices (see Figure~\ref{fig5}(b)). 
Let $b_z,b_{z+1}\ldots b_x$ and $b_x,b_{x+1},\ldots,b_y$ be the neighbours of $a_i$ and $a_j$ in $C$ respectively in clockwise order. 
Let $a_k$ be an $\alpha$ vertex in $C$. By Observation~\ref{Obs:nbeta2}, $a_i$, $a_j$ and $a_k$ must be the only vertices in $C$ that are not $\gamma$ vertices.  
Hence $b_y$ and $b_z$ must be neighbours of $a_k$ and $b_yb_z$ must be an edge in $C'$.      

As $a_i$ has two neighbours in $C'$ which are not $\delta$ vertices ($b_x$ and $b_z$), we conclude using part (1) 
of Observation~\ref{obs:gamma} that $b_y$ cannot be a neighbour of $a_i$. Similarly, $b_z$ cannot be a neighbour of $a_j$.  
Moreover, there cannot be a chord between $b_x$ and $b_y$ or between $b_z$ and $b_x$ (by Lemma~\ref{lem:multiC}). 
Consequently, the vertices $a_i,b_x,a_j,b_y, b_z$ should induce an odd hole which contains vertices only in $C$ and $Int(C)$, 
a contradiction. Hence at least one among $a_i$ and $a_j$ must be an $\alpha$ vertex.
\end{proof}
\begin{lemma}
If $N_\alpha=0$, then there exist a $\beta'$ vertex $b_i$ and a $\delta$ vertex $b_x$ satisfying the following:
\begin{itemize}\label{lem:nalpha0}
\item[(a)] $b_i$ is the unique $\beta'$ vertex in $C'$.
\item[(b)]$b_ib_x$ is a chord in $C'$.
\item[(c)]Every vertex in $C$ is adjacent to either $b_i$ or $b_x$.
\end{itemize}
\end{lemma}
\begin{proof}
\begin{itemize}
\item[(a)] By Corollary~\ref{nbeta1}, there exists at least one $\beta'$ vertex in $C'$. 
Also by Corollary~\ref{cor:nbeta1}, there exists at least one $\beta$ vertex in $C'$. 
As $N_\beta + N_{\beta'}=3$ (by Observation~\ref{Obs:nbeta2}), it is enough to prove that $N_{\beta'} \neq 2$. 
Assume that $N_{\beta'} = 2$. Let $b_i$ and $b_j$ be two consecutive $\beta'$ vertices in clockwise order in $C$. 
Then  $b_i$ and $b_j$ must have a common neighbour (say $a_k$) in $C$. As $N_\alpha=0$ and $N_\alpha+N_{\alpha'}>0$, 
$a_k$ must be an $\alpha'$ vertex. But this is not possible by Lemma~\ref{lem:alphabeta1}. Therefore, $N_{\beta'} = 1$.
\item[(b)]  By Observation~\ref{Obs:nbeta2} and part (a), $N_\beta=2$. Let $b_j$ and $b_k$ be the $\beta$ vertices and $b_i$ be the $\beta'$ vertex in $C'$. 
Let $a_p,a_{p+1},\ldots,a_q$ ($q \geq p+3$) be the neighbours of $b_i$ in $C$ arranged in clockwise order (see figure~\ref{fig6}(a)). 
Let $a_q$ and $a_r$ be the neighbours of $b_j$ and let $a_r$ and $a_p$ be the neighbours of $b_k$ in $C$ in clockwise order. Note that $a_p,a_q$ and $a_r$ are 
$\alpha'$ vertices by assumption. Hence  no two vertices from the set $\{b_i, b_j, b_k\}$ are consecutive on $C'$. Note that $b_j,b_{j+1},\ldots,b_k$ are  the neighbours of $a_r$ in $C'$ Then the vertices 
$a_p,b_i,a_q,b_j,b_{j+1},\ldots,b_k,a_p$ forms an 
odd cycle (say $C''$).
Since $C''$ contains vertices only in $C$ and $Int(C)$, it cannot be an odd hole.
Hence, there must be at least one chord inside $C''$. However, there is no chord between any two vertices in $\{b_i,b_j,b_k\}$ 
and between any two vertices in $\{b_j,b_{j+1},\ldots,b_k\}$ (by Observation~\ref{Obs:triangle}).  
Hence the only possibility for the chord is $b_ib_x$ such that $x\in\{j+1,\ldots,k-1\}$ (see Figure~\ref{fig6}(b)). 
\item[(c)] By Part (a), since $b_i$ is the unique $\beta'$ vertex in $C'$, every $\gamma$ vertex in $C$ must be adjacent to $b_i$ (part 2 of Observation~\ref{obs:gamma}). Moreover as $N_\alpha=0$, $b_i$ must have two $\alpha'$ vertices (say $a_p$ and $a_q$) in $C$ as neighbours. Consequently, as $N_{\alpha'}=3$ (Observation~\ref{Obs:nbeta2}), the only one vertex that is not a neighbour of $b_i$ in $C$ is an $\alpha'$ vertex, say $a_r$ (see Figure~\ref{fig6}(b)). By part (b) there exists a chord $b_ib_x$ between the $\beta'$ vertex $b_i$ and a $\delta$ vertex $b_x$ in $C'$. As $b_x$ is a $\delta$ vertex, it has exactly one neighbour on $C$. This neighbour cannot be $a_p$ (or $a_q$) as otherwise the vertices $b_i,b_x$ and $a_p$ (respectively $b_i,b_x$ and $a_q$) will form a separating triangle in $G$. Hence $a_r$ is the neighbour of $b_x$ in $C$.
\end{itemize}
\end{proof}
\begin{Corollary}\label{cor:alpha0}
If $N_\alpha=0$, then there exist a chord $b_ib_x$ in $C'$ such that the boundary of the exterior face of the local neighbourhood of $b_ib_x$ is the odd hole $C$.
\end{Corollary}
Note that Corollary~\ref{cor:alpha0} gives a local characterization for the odd hole $C$ when
$N_{\alpha}=0$. That is, there exists an edge in $G$, the exterior boundary of its local neighbourhood is the hole $C$.   
Our goal is to obtain a similar local characterization when $N_{\alpha}>0$.     

\begin{figure}[!htb]
\centering
 \includegraphics[scale=0.70]{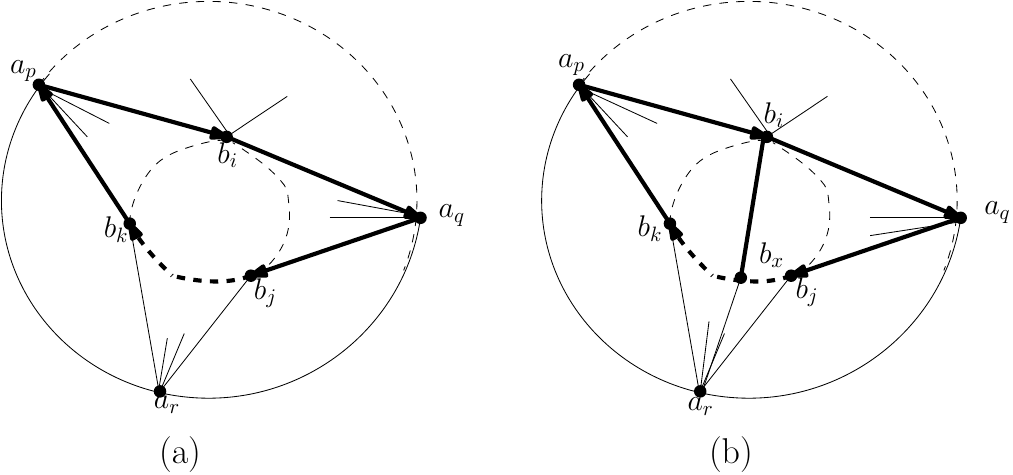}
\caption{$N_{\alpha'}=3$ }
\label{fig6}
\end{figure}

For the rest of the paper, we use the following notation.
Let $A$ be the set of all vertices in $C$ of type $\alpha$ or $\alpha'$ and
let $B$ be the set of all vertices in $C'$ of type $\beta$ or $\beta'$.  
By Observation~\ref{Obs:nbeta2}, $|A|=|B|=3$.   
Let $a_i,a_j$ and $a_k$ (respectively $b_x,b_y$ and $b_z$) be the vertices of $A$ (respectively $B$) listed in clockwise order in $G'$. 
There exists at least one $\beta$ vertex in $B$ (by Corollary~\ref{cor:nbeta1}) and at least one $\alpha'$ vertex in $A$ (by Corollary~\ref{cor:nalpha1}). 
We fix $b_y$ to be a $\beta$ vertex in $B$  and $a_i$ to be an $\alpha'$ vertex in $A$. 
The next lemma shows that every vertex in $C$ must be a neighbour of either $b_x$ or $b_z$ (or both).  
\begin{lemma}\label{lem:betaminus}
Let $b_y\in B$ be a $\beta$ vertex.  Then the boundary of the exterior 
face of the subgraph induced by the closed neighbourhood of $B\setminus \{b_y\}$
 in $G$ is the odd hole $C$.   
\end{lemma}
\begin{proof}
Since $b_y$ is a $\beta$ 
vertex, $b_y$ must have exactly two neighbours in $C$ - say $a_i$ and $a_j$ 
in clockwise order. Further, we have $j=i+1$ and $a_ia_j$ will be an edge in $C$.  
(See Figure~\ref{fig7}(a)).    
Moreover, $a_i$ and $a_j$ cannot have another common neighbour $b_t$
for any $t\in \{x,z\}$ as otherwise $b_t,a_i$ and $a_j$ will form a separating
triangle with $b_y$ in the interior.  As each of $a_i$ and $a_j$ 
must have a neighbour in $B$ distinct from $b_y$ (by Observation~\ref{obs:gamma}), 
we may assume without loss of generality that $b_x$ is a neighbour
of $a_i$ and $b_z$ is a neighbour of $a_j$.  Thus, each neighbour of $b_y$ in $C$
is either a neighbour of $b_x$ or a neighbour of $b_z$.  
Since every vertex in $C$ must be
a neighbour $b_x$ or $b_y$ or $b_z$ (Lemma~\ref{obs:Nalphabeta}), it follows that every
vertex in $C$ is a neighbour of $b_x$ or $b_z$.  The lemma follows since 
$C$ is assumed to be the boundary of the exterior face of the subgraph induced 
by the closed neighbourhood of $B$.      
\end{proof}
Since Corollary~\ref{nbeta1} guarantees that $C'$ contains a $\beta$ vertex, 
we now have a local characterization for the odd hole $C$
in the sense that $C$ will be the exterior boundary of the local neighbourhood of the two element set 
$\{b_x,b_z\}$. We now strengthen Lemma~\ref{lem:betaminus} by showing that there exists 
an edge in $G$ (connecting two vertices in $C'$) whose local neighbourhood has $C$ as 
its exterior boundary.  
\begin{lemma}\label{lem:onebeta}
 There exists two vertices $b_p$ and $b_q$ in $C'$ such that:
 \begin{itemize}
 \item  $b_pb_q$ is an edge  in $G$.
 \item The boundary of the exterior face of the closed neighbourhood of $b_pb_q$ is the odd hole $C$.
 \end{itemize}
\end{lemma}
\begin{proof}
If $N_\alpha=0$, then the result follows by Corollary~\ref{cor:alpha0}. Hence, assume that $N_\alpha>0$.

Recall that we have assumed $b_y$ to be a $\beta$ vertex.  
Therefore, if $b_x$ and $b_z$ are adjacent in $C'$, setting $b_pb_q=b_x b_z$ suffices (by Lemma~\ref{lem:betaminus}). 

If $b_x b_z$ are not adjacent in $C'$,
the common neighbour of $b_x$ and $b_z$ on $C$ must be an $\alpha'$ vertex. By Lemma~\ref{lem:alphabeta1},
either $b_x$ or $b_z$ is a $\beta$ vertex and by Corollary~\ref{nbeta1}, at least one of them must be a $\beta'$ vertex.
Hence, without loss of generality, we may assume that $b_x$ is a $\beta$ vertex and $b_z$ is a $\beta'$ vertex. 
By Lemma~\ref{lem:betaalpha1}, at least one of the neighbours of $b_z$ must be an $\alpha$ vertex. Therefore,
the common neighbour of $b_z$ and $b_y$ in $C$ must be an $\alpha$ vertex. This implies that $b_z b_y$ is an
edge of $C'$ (see Figure~\ref{fig7}(b)). Since $b_x$ is a $\beta$ vertex, every neighbour of $b_x$ in $C$ must be a neighbour of $b_y$ or $b_z$. Hence, by Lemma~\ref{lem:betaminus}, setting $b_pb_q=b_z b_y$ 
suffices to complete the proof. 
\end{proof}
\begin{figure}
\centering
 \includegraphics[scale=0.75]{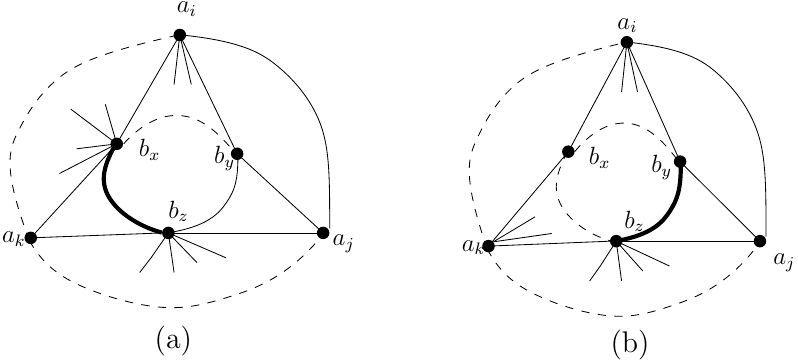}
\caption{ (a)$N_{\beta}=1$, $N_{\alpha'}=1$ (b)$N_{\beta}=2$, $N_{\alpha'}=2$ and $b_xb_z$ is not an edge in $C'$}
\label{fig7}
\end{figure}
Let us now consolidate our observations so far. 
Recall that our objective was to prove Theorem~\ref{thm:main-thm}, to obtain 
a local characterization for a plane near-triangulation to be perfect. Since the graph can be decomposed into induced $2$-connected subgraphs containing 
no separating triangles or edge separators without affecting the perfectness of the graph, it was sufficient to limit the attention to  
W-triangulations.   We noted that if any internal vertex $x$ in a W-triangulation has on odd degree, 
the result was immediate because the exterior face of the closed neighbourhood of $x$ would have
been an odd hole. Consequently, the core of the problem was to characterize odd holes in even W-near-triangulated (induced) 
subgraphs of the original plane near-triangulation. If an even W-near-triangulation $G$ is non-perfect, we considered a \textit{minimal odd hole} $C$ 
in $G$ such that there is no other odd hole in $C \cup Int(C)$.
Then we considered the cycle $C'$ forming the boundary of the subgraph $H$ induced by the vertices
in $Int(C)$. We showed that either (i) $C'$ is a (non-separating) triangle, the exterior face of the closed neighbourhood of which is the odd hole 
$C$ or (ii) there exists an edge $b_pb_q$ connecting vertices in $C'$ such that the boundary of the exterior face of the 
closed neighbourhood of the edge $b_pb_q$ is the odd hole $C$ (Lemma~\ref{lem:onebeta}). Thus we have:  
\begin{Corollary}\label{evenwtriangulation}
A W-triangulation $G$ is not-perfect if and only if $G$ contains at least one among the following:
 \begin{itemize}
 \item a vertex $x$ (of odd degree), the exterior boundary of the local neighbourhood of which, is an odd hole.
 \item an edge $xy$, the exterior boundary of the local neighbourhood of which, is an odd hole.
 \item a facial triangle $xyz$, the exterior boundary of the local neighbourhood of which, is an odd hole.
 \end{itemize}
\end{Corollary}

From this, Theorem~\ref{thm:main-thm} is immediate. 

We make note of a few details about Theorem~\ref{thm:main-thm} which are of significance while translating
Theorem~\ref{thm:main-thm} into an algorithm for checking whether a plane near-triangulation is perfect.      
Let $G$ be a plane near-triangulation that is not perfect.  Let  $X=xyz$ be a facial triangle in an induced even W-near-triangulated 
subgraph $G'$ of $G$, such that exterior boundary of the local neighbourhood of $X$ induces an odd hole in $G$.  
It must be noted that $X$ need not necessarily be a facial triangle in the original graph $G$, but could be a separating triangle in $G$.
Hence, identification of the separating triangles in $G$ becomes a significant algorithmic consideration.  It turns out this task
is easy due to the algorithm by \citet{kant1997more} that identifies the separating triangles and the  W-triangulated
subgraphs of $G$ in linear time.  

The second point to note is the following.  We \textit{cannot} conclude from Theorem~\ref{thm:main-thm} that in a non-perfect near-triangulation 
$G$, we can find an induced subgraph $X$ which is either a vertex, an edge or a triangle, such that the neighbours of $X$ induces an odd
hole in $G$.  That is, it is necessary to inspect the \textit{exterior boundary} of the local neighbourhood of each vertex, edge and face 
to detect an odd hole.  Figure~\ref{fig9} gives an example of a graph illustrating this fact.  In this example, the exterior 
boundary of the local neighbourhood of the edge $ip$ induces the odd hole $abcde$.  However, the subgraph induced by the neighbours 
of the edge $ip$ (or any  other edge or face) does not induce an odd hole.  The requirement to check the exterior boundary of 
the local neighbourhood of each edge and each face results in quadratic time complexity for the algorithm described in the next section.
If a W-near-triangulation $G$ does not contain any induced wheel on five vertices,
then it is known that $G$ is not perfect if and only if 
it contains a vertex or a face whose neighbours induce an 
odd hole~\cite{SalamCWKS19}.   
Consequently, for this restricted class of graphs, 
checking perfectness requires only sub-quadratic time.  
\begin{figure}
\centering
 \includegraphics[scale=0.7]{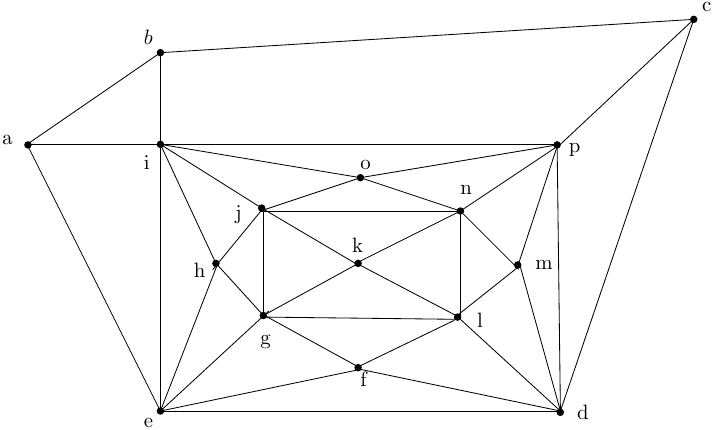}
\caption{The neighbours of the edge $ip$ does not induce an odd hole though the boundary of its local neighbourhood induces
an odd hole.}
\label{fig9}
\end{figure}

\section{Recognition of perfect plane near-triangulations}\label{perfect plane}
In this section, we describe an $O(n^2)$ algorithm to determine
whether a given plane near-triangulation $G$ of $n$ vertices 
is perfect.  The algorithm essentially checks the conditions of Theorem~\ref{thm:main-thm},
using well known techniques for handling planar graphs.  For the sake of clarity, 
the steps of the algorithm and its analysis are presented in detail.  

Initially, we assume that $G$ is a plane triangulation.  Later, we will
describe how to handle near-triangulations as well.  

Given a plane triangulation $G$, using the linear time algorithm of \citet{kant1997more} for identifying separating triangles in a plane triangulation, 
$G$ can be split into its $4$-connected blocks.  It is well known that a plane 
triangulation is $4$-connected 
if and only if it is free of separating triangles.  Consequently, 
the $4$-connected blocks identified by Kant's algorithm are precisely 
the maximal induced W-triangulations in $G$.  We will call each such $4$-connected block a
{\em W-component} of $G$ \cite{SalamCWKS19}.   We have already seen that $G$ is perfect 
if and only if all its W-components are
perfect. Hence, our task reduces to the problem of finding a quadratic 
time algorithm for 
checking whether a given plane triangulation without separating triangles 
is perfect. We split our task into four stages, as described below.
\subsection{Preprocessing}\label{preprocess}
Let $H$ be a plane triangulation on $n$ vertices, without separating triangles. 
Suppose we are given the adjacency list of $H$ as input. 
We will first compute a planar straight line embedding of $H$ on a $n \times n$ grid, 
using the linear time algorithm of \citet{schnyder1990embedding}. This essentially gives us the $(x, y)$ coordinates of each vertex of $H$.  
Using these coordinates, the slopes of all the edges of $H$ can be computed 
in $O(m)$ time, where $m$ is the number of edges of $H$. Note that,
for a planar graph, $m=O(n)$.  It would be
useful to preprocess the adjacency list of $H$ at this point, so that for each vertex $v$, its neighbours are arranged in the descending order of the 
slopes of edges incident at $v$. This preprocessing can be done in $O(m \log{n})$ time.  
The next three stages involve verifying the conditions of perfectness mentioned in Theorem~\ref{thm:main-thm} one by one. 
\subsection{Checking vertex degrees}\label{vertexdegree}
If any vertex $u$ of $H$ is of odd degree, the open neighbourhood $N_{H}(u)$ induces an odd hole
and $H$ is not perfect. Checking the degrees of all vertices can be easily done in time linear in $n$. 
If every vertex of $H$ is of even degree, we need to check the remaining conditions to verify the perfectness of $H$.  
\subsection{Checking the boundary of the local neighbourhood of each edge} \label{sec:neighbourhoodedge}
The next step is to check if the boundary of the local neighbourhood of any edge of $H$ is an odd hole. 
For each edge $uv$ of $H$, define $S_{uv}=N(u) \cup N(v) \cup \{u, v\}$ and construct an indicator array $A_{uv}$ of length $n$, 
where $A_{uv}[i]=1$ if vertex $v_i \in S_{uv}$ and zero otherwise. The construction of these arrays can be done
in $O(mn)=O(n^2)$ time in total. 

Now, we construct $m$ plane subgraphs of $H$, one corresponding to each edge of $H$. 
The graph $H_{uv}$ will be the induced subgraph of $H$ on the vertex set $S_{uv}$. 
We store the coordinate position information of each vertex of $H_{uv}$, by copying the same from $H$. 
To get the sorted adjacency list of $H_{uv}$, we start with a copy of the sorted adjacency list of $H$. Then, 
mark the lists of vertices not in $S_{uv}$ as deleted. For each vertex $x$ in $S_{uv}$, go through the adjacency list
$x$ in order, and when an edge $xv_i$ such that $A_{uv}[i]=0$ is encountered, then mark the edge as deleted. 
It takes only $O(m+n)$ time for obtaining the sorted adjacency list of $H_{uv}$, in this manner. 
The construction of subgraphs $H_{uv}$ corresponding to each
edge $uv$ of $H$ along with the information mentioned above, takes only $O(m (m+n))=O(n^2)$ time in total. 

Now, we describe a procedure to check if the exterior boundary of an induced subgraph $H_{uv}$ constructed above is an odd hole. 
It is easy to see that the exterior boundary of $H_{uv}$ is an induced cycle. So, it suffices to identify the vertices on the exterior boundary and
check the parity of their count. Since we have coordinates of vertices from a straight line drawing, this is easy. First, identify
the left most vertex of $H_{uv}$ in the straight line embedding. This only involves identifying a vertex of $H_{uv}$ with the smallest $x$ coordinate.
Clearly, this vertex is on the exterior face of $H_{uv}$. Let this vertex be $l$. Recall that the adjacency list of each vertex is 
sorted in the decreasing order of slopes. The edge with the largest slope incident at vertex $l$ must be on the exterior boundary of $H_{uv}$.
The other endpoint of this edge can be identified as the next vertex on the boundary of $H_{uv}$. 
After identifying a new vertex $x$ on the boundary, it is easy to identify the next one.  Suppose W is the vertex identified before $x$. 
If $y$ is the vertex that
succeeds W (in the cyclic order) in the adjacency list of $x$, then $y$ is the next vertex on the boundary of $H_{uv}$. When this procedure encounters the 
initial vertex $l$ again, we would have identified all the vertices on the boundary of $H$. Thus, identifying the exterior 
boundary of $H_{uv}$ and checking whether it
is an odd hole, can be done in $O(m+n)$ time. For checking the boundaries of all subgraphs $H_{uv}$ we constructed, total 
time required is only $O(m (m+n))=O(n^2)$. 

If this check fails to find an odd hole, we have to check the boundary of 
the local neighbourhood of each triangle, as described below.  
\subsection{Checking the boundary of the local neighbourhood of each triangle}\label{sec:neighbourhoodtriangle}
Since $H$ is free of separating triangles, triangles of $H$ are precisely 
its faces.  
A listing of all the faces of $H$ can be done in linear time, 
by traversing the adjacency list of every
vertex once. Checking whether the exterior boundary of the closed neighbourhood of a triangle $uvw$ forms an odd hole can be
done in a way very similar to the method we discussed in the previous subsection. For each triangle $uvw$, we will define
a set $S_{uvw}=N(u) \cup N(v) \cup N(w) \cup \{u, v, w\}$ and an indicator array $A_{uvw}$ for $S_{uvw}$. Then, we can construct
plane subgraphs $H_{uvw}$, the induced subgraph of $H$ on the vertex set $S_{uvw}$. The method of checking whether the boundary of 
$H_{uvw}$ is an odd hole or not, is the same as the method described earlier for $H_{uv}$. 
The number of subgraphs to be processed is the number of faces of $H$, which is linear in $n$. Hence, the time required for
checking the boundaries of each such subgraph is again $O(n^2)$ only.
\subsection{Handling plane near-triangulations}
The method described above for recognizing planar perfect graphs can be 
extended to recognize perfect plane near-triangulations by the simple
modifications described below, without affecting the complexity of the algorithm.   
The strategy is to triangulate the plane near-triangulation; use the algorithm of \citet{kant1997more} to 
identify the W-components of the triangulated graph; and retrieve the W-components of 
the original plane near-triangulation.  

Let $G(V,E)$ be a plane near-triangulation.  If $G$ is not $2$-connected,
in linear time we can find the $2$-connected components of $G$ using depth first search, and work on each component.
Hence, we assume without loss of generality that  $G$ is $2$-connected.   

We can embed $G$ into an $n\times n$ grid using the algorithm of \citet{schnyder1990embedding} 
in $O(n)$ time, preprocess the adjacency list of $G$ as described in Subsection~\ref{preprocess} in $O(m\log n)$ time
and identify the vertices on the boundary of exterior face of $G$ by using the method discussed in 
Subsection~\ref{sec:neighbourhoodedge} in $O(n)$ time. Let $v_1,v_2,\ldots,v_k,v_1$ be the cycle forming the  
the boundary of exterior face of $G$. 
Construct a plane triangulation $G'$ from $G$ by adding a new vertex $v_0$ on the exterior face of $G$ and adding edges from $v_0$ to $v_i$  
for every $i \in \{1,2,\ldots,k\}$. The adjacency list of $G$ can be modified to get the adjacency list of $G'$ in linear time. 
Note that, as $G'$ is a triangulation, it contains no edge separators.
We can use the linear time algorithm of \citet{kant1997more} for identifying separating triangles 
in the plane triangulation $G'$ and decompose $G'$ into W-components.

Note that $v_iv_j$ is an edge separator of $G$ if and only if it is a 
chord connecting two vertices on the external face of $G$, forming a separating triangle
$v_iv_jv_0$ in $G'$. Conversely, a separating triangle in $G'$ containing $v_0$ must be of the form 
$v_0v_iv_j$ for some $i,j\in \{1,2,\ldots, k\}$, with $v_iv_j$ forming a chord connecting two vertices on 
the exterior face of $G$. Moreover,  every separating triangle in $G$ will be a separating triangle in 
$G'$ and every separating triangle in $G'$ that does not contain $v_0$ is a separating triangle in $G$.  
Hence, we have the following observation.    

\begin{Observation}\label{obs:reduction}
Let $G_1',G_2',\ldots , G_r'$ be the W-components of $G'$. Let $G_i=G_i'$ if 
$v_0\notin V(G_i')$ and $G_i=G_i'\setminus\{v_0\}$ otherwise.  Then, 
$G_1,G_2,\ldots ,G_r$ are precisely the W-components of $G$.  
\end{Observation}


As a consequence of Observation~\ref{obs:reduction}, it suffices
to decompose $G'$ into its W-components and remove the vertex $v_0$ from each W-component to recover the W-components of $G$. 
Removing $v_0$ from all the W-components of $G'$ 
requires only traversing the adjacency lists of all the graphs once
and hence possible in $O(n)$ time. It is easy to verify that 
the $O(n^{2})$ method described earlier for checking perfectness of the W-components
of a plane triangulation suffices for handling the W-components of the plane
near-triangulation $G$ as well.  Thus, it follows that checking the perfectness of
plane near-triangulations requires only $O(n^{2})$ time.  
\section{Acknowledgment}
We thank Sunil Chandran, IISc. Bangalore and Ajit A. Diwan, IIT Bombay 
for discussions and suggestions. We thank the latter also for the example
in Figure~\ref{fig9}.  
\bibliographystyle{plainnat}

\end{document}